\documentclass[runningheads]{llncs}
\usepackage[T1]{fontenc}
\usepackage{graphicx}
\usepackage{color}
\usepackage[hidelinks]{hyperref}
\usepackage{amsmath}
\usepackage{booktabs}
\usepackage{xcolor}
\usepackage{multirow}
\usepackage{float}

\newcommand{\eps}{\varepsilon}

\newcommand{\calA}{\mathcal{A}}
\newcommand{\calB}{\mathcal{B}}
\newcommand{\calI}{\mathcal{I}}

\newcommand{\POF}{\textnormal{\textsf{POF}}}
\newcommand{\EW}{\textnormal{\textsf{EW}}}
\newcommand{\MEW}{\textnormal{\textsf{MEW}}}
\newcommand{\EFi}{\textnormal{\textsf{EF1}}}
\newcommand{\RR}{\textnormal{\textsf{RR}}}
\newcommand{\Ba}{\textnormal{\textsf{Ba}}}
\newcommand{\MUW}{\textnormal{\textsf{MUW}}}
\newcommand{\MNW}{\textnormal{\textsf{MNW}}}

\newcommand{\EG}{\textnormal{EG}}

\spnewtheorem{fact}[theorem]{Fact}{\bfseries}{\itshape}
\spnewtheorem{defn}[theorem]{Definition}{\bfseries}{\itshape}
\spnewtheorem{thm}[theorem]{Theorem}{\bfseries}{\itshape}
\spnewtheorem{corl}[theorem]{Corollary}{\bfseries}{\itshape}
\spnewtheorem{lem}[theorem]{Lemma}{\bfseries}{\itshape}

\begin{document}
\title{Egalitarian Price of Fairness for Indivisible~Goods\thanks{A shorter version of this paper is published in the Proceedings of 20th Pacific Rim International Conference on Artificial Intelligence (PRICAI) at \url{https://doi.org/10.1007/978-981-99-7019-3_3}.}}
%
%
\author{Karen Frilya Celine\inst{1}\orcidID{0000-0002-7078-5582} \and
Muhammad Ayaz Dzulfikar\inst{1}\orcidID{0009-0002-7962-0677} \and
Ivan Adrian Koswara\inst{1, 2}\orcidID{0000-0002-9311-6840}}
\authorrunning{K. F. Celine, M. A. Dzulfikar and I. A. Koswara}
%
\institute{School of Computing, National University of Singapore, Singapore\\
\email{\{karen.celine,ayaz.dzulfikar\}@u.nus.edu, ivanak@comp.nus.edu.sg}
\and {Corresponding author}}
\maketitle              
\begin{abstract}
In the context of fair division, the concept of price of fairness has been introduced to quantify the loss of welfare when we have to satisfy some fairness condition. In other words, it is the price we have to pay to guarantee fairness. Various settings of fair division have been considered previously; we extend to the setting of indivisible goods by using egalitarian welfare as the welfare measure, instead of the commonly used utilitarian welfare. We provide lower and upper bounds for various fairness and efficiency conditions such as envy-freeness up to one good (EF1) and maximum Nash welfare (MNW).
\end{abstract}

\section{Introduction}

Fair division is the problem of allocating scarce resources to agents with possibly differing interests. It has 
many real world applications, such as the distribution of inheritance, divorce settlements and airport traffic management. Economists have studied fair division as far back as the 1940s~\cite{DubinsSpanier1961,Steinhaus1948}. Recently, the problem of fair division has also received significant interest in artificial intelligence~\cite{AmanatidisEtAl2018,ijcai2019p0012,OhEtAl2019}.

In a fair division problem, there are several possible goals to strive for. One goal is \textit{fairness}, where each individual agent should feel they get a fair allocation; another is \textit{social welfare}, where the goal is to optimize the welfare of all agents as a whole.
These goals are not always aligned. For example, to maximize the sum of utilities of the agents (i.e. utilitarian welfare), the optimal allocation is to assign each item to the agent that values it the most. Clearly this allocation can be far from fair, as an agent might be deprived of every item. However, making the allocation fairer comes at the cost of decreasing the total welfare. In other words, there is a price to pay if we want a division to be fair.

The notion of \textit{price of fairness} was introduced independently by Bertsimas et al.~\cite{pof} and Caragiannis et al.~\cite{Caragiannis2012} to capture this concept. Initially, the setting was for utilitarian welfare on divisible goods. Since then, there have been other works discussing the setting of utilitarian welfare with indivisible goods~\cite{Barman2020,ijcai2019p0012}, as well as the setting of egalitarian welfare with divisible goods~\cite{AumannDombb2015,Caragiannis2012}. Since the same cannot be said for egalitarian welfare with indivisible goods, our paper completes the picture by investigating this setting.

One problem with investigating fairness conditions is that they might not have a satisfying allocation for some instances, especially when the goods are indivisible. We follow the method in Bei et al.~\cite{ijcai2019p0012} of handling this problem by considering only fairness conditions which can always be satisfied in all instances for any number of agents. As such, we do not investigate properties such as envy-freeness and proportionality, which are not guaranteed to be satisfiable. Special cases such as envy-freeness up to any good (EFX) which has been shown to be satisfiable for $n \le 3$ agents can be considered for future works.

We study the price of fairness of three fairness properties: envy-freeness up to one good (EF1), balancedness, and round-robin. Not only are these properties always satisfiable, but an allocation which has all three properties can be easily found by the round-robin algorithm. Furthermore, these fairness notions are widely studied in the literature. In particular, tight bounds for the utilitarian price of fairness of these properties have been found~\cite{ijcai2019p0012}, which allows for comparison between the utilitarian and egalitarian prices of fairness.

Moreover, we also study the price of fairness of two welfare maximizers: maximum utilitarian welfare (MUW) and maximum Nash welfare (MNW). While these are efficiency notions instead of fairness notions, they are crucial to the study of resource allocation. Studying their prices of fairness helps us compare between the different types of welfare maximizers, and might shed light on if and when one type of welfare function would best quantify social welfare.

\subsection{Our Results}

We investigate the upper and lower bounds of the price of fairness for five fairness and efficiency properties described above. Letting $n$ be the number of agents in the instance, we show that EF1, balancedness, and round-robin have price of fairness $\Theta(n)$. Meanwhile, MUW and MNW have infinite price of fairness, except for the case of MNW with $n = 2$ where the price of fairness is finite. Our results are summarized in Table~\ref{tab:summary}. We have also included the utilitarian prices of fairness found by Bei et al.~\cite{ijcai2019p0012} for comparison. We restrict our attention to the general instances for any fixed $n$; future work can be done on specializing to, say, instances with identical ordering, or some other constraint, in case it can bring down the price of fairness for some of the properties.

\begin{table*}[ht]
\caption{Summary of results}
\label{tab:summary}
\centering
\begin{tabular}{|l|c|c|c|}
    \hline
    \multicolumn{2}{|c|}{\multirow{2}{*}{Property}} & \multicolumn{2}{c|}{Price of fairness} \\
    \cline{3-4}
    \multicolumn{2}{|c|}{} & {Egalitarian} & {Utilitarian~\cite{Barman2020,ijcai2019p0012}} \\
    \hline
    \multicolumn{2}{|l|}{Envy-free up to one good (EF1) } & $\Theta(n)$ & $\Theta(\sqrt{n})$\\
    \hline
    \multicolumn{2}{|l|}{Balanced} & $n$ & $\Theta(\sqrt{n})$\\
    \hline
    \multicolumn{2}{|l|}{Round-robin algorithm (RR)} & $\Theta(n)$ & $n$ \\
    \hline
    \multirow{2}{*}{Maximum Nash welfare (MNW)} & ($n = 2$) & $\approx 2$ & $\approx 1.2$ \\
    \cline{2-4}
    & ($n \ge 3$) & $\infty$ & $\Theta(n)$ \\
    \hline
    \multicolumn{2}{|l|}{Maximum utilitarian welfare (MUW)} & $\infty$ & 1 \\
    \hline
    \multicolumn{2}{|l|}{Maximum egalitarian welfare (MEW)} & 1 & $\Theta(n)$ \\
    \hline
\end{tabular}
\end{table*}

In a way, our results are surprising compared to the utilitarian results. Utilitarian welfare is purely an efficiency notion, while egalitarian welfare captures some sort of ``fairness'', since maximizing the utility of the poorest agent means that every agent's utility is taken into consideration and no agent's poverty can be ignored. However, the egalitarian price of fairness for the properties are actually worse (higher) than the utilitarian price of fairness. Despite appearing ``fairer'', egalitarian welfare turns out to be less fair when we impose other fairness conditions.

\subsection{Related Work}

As mentioned above, Bertsimas et al.~\cite{pof} and Caragiannis et al.~\cite{Caragiannis2012} independently introduced the concept of \textit{price of fairness}. Bertsimas et al. studied it in the context of divisible goods, while Caragiannis et al. studied both goods and chores whether they are divisible or indivisible.  Since then, the price of fairness has been studied in other settings. In the context of contiguous allocations,
the price of fairness has been studied for divisible goods~\cite{AumannDombb2015}, indivisible goods~\cite{Suksompong2019}, divisible chores~\cite{HeydrichVanStee2015} as well as indivisible chores~\cite{HohneVanStee2021}. Li et al.~\cite{LiEtAl2022} studied the price of fairness of almost weighted proportional allocations for indivisible chores.
Additionally, Bil\`o et al.~\cite{Bilo2016} studied it in the context of machine scheduling, while Michorzewski et al.~\cite{Michorzewski2020} studied it in the context of budget division.

Typically, the price of fairness refers to the utilitarian price of fairness which measures the loss of utilitarian welfare due to fairness constraints. However, the price of fairness can also be defined with respect to other social welfare functions. For example, in the context of egalitarian welfare, Aumann and Dombb~\cite{AumannDombb2015} and Suksompong~\cite{Suksompong2019} studied the price of fairness for contiguous allocations of divisible and indivisible goods respectively. More generally, Arunachaleswaran et al.~\cite{arunachaleswaran2021fair} used the generalized H\"older mean with exponent $\rho$ as their welfare function. In particular, when $\rho = 1, 0, -\infty$, the generalized mean corresponds to utilitarian, Nash, and egalitarian welfare respectively. This is done in the context of approximately envy-free allocations of divisible goods.

Most studies express the price of fairness as a function of the number of agents $n$. However, there are cases where the price of fairness (for indivisible goods) depends also on the number of goods $m$. Kurz~\cite{Kurz2014} studied the price of envy-freeness in terms of both the number of agents and the number of goods, and showed that when the number of goods is not much larger than the number of agents, the price of fairness can be much lower. Bei et al.~\cite{ijcai2019p0012} proved a similar result for round-robin allocations.

More generally, fair division has been an active area of research, with many studies investigating different ways to define fairness, including envy-freeness up to one good (EF1), envy-freeness up to any good (EFX), maximin share (MMS), and pairwise maximin share (PMMS) \cite{AmanatidisEtAl2018,BiswasBarman2018,Caragiannis2019,GhodsiEtAl2018,KurokawaEtAl2018b,KyropoulouEtAl2020,Markakis2017,OhEtAl2019,PlautRoughgarden2020}. Many of these focus on the setting of indivisible goods.

\section{Preliminaries}

An \textbf{instance} $\calI$ consists of the agents $N = \{1, 2, \ldots, n\}$, the (indivisible) goods $M = \{1, 2, \ldots, m\}$, and each agent's utility function $u_i$. We assume $n \ge 2$. The utility function is \textit{nonnegative}, i.e. $u_i(j) \ge 0$ for all $i,j$. It is \textit{additive}, i.e. $u_i(A) = \sum_{j \in A} u_i(j)$ for a set of goods $A$. It is also \textit{normalized}, i.e. $u_i(M) = 1$, so that each agent values the whole bundle identically.

An \textbf{allocation} $\calA$ for an instance is a partition $(A_1, \ldots, A_n)$ of the goods $M$ such that agent $i$ receives bundle $A_i$. A \textbf{property} $P$ is a Boolean predicate on the allocations; alternatively, it maps each instance $\calI$ to the set $P(\calI)$ of allocations satisfying the property. A property is \textbf{always satisfiable} if $|P(\calI)| \ge 1$ for all $\calI$.

The \textbf{egalitarian welfare} of an allocation $\calA$ of an instance $\calI$ is
$$\EW(\calI, \calA) := \min_{i \in N} u_i(A_i).$$

The \textbf{maximum egalitarian welfare (MEW)} (also \textbf{optimal welfare}) of an instance $\calI$ is the highest possible egalitarian welfare for that instance; it is denoted $\MEW(\calI)$. Its \textbf{optimal $P$ welfare} only considers allocations that satisfy property $P$; it is denoted $\MEW_P(\calI)$. An allocation achieving the MEW is also said to satisfy property MEW.

\begin{defn}[Price of fairness]
The \textbf{price of fairness (POF)} of a property $P$ for instance $\calI$ is

$$\POF_P(\calI) := \frac{\max_{\mathcal{A}} \EW(\mathcal{A})}{\max_{\mathcal{A} \in  P(I)} \EW(\mathcal{A)}} = \frac{\MEW(\calI)}{\MEW_P(\calI)}.$$

For price of fairness, we use the convention $0/0 = 1$ and $x/0 = \infty$ for $x > 0$.

The price of fairness of a property $P$ over a family of instances is the supremum of the price of fairness over those instances.
\end{defn}

Price of fairness is traditionally represented as a function in terms of the number of agents $n$. We follow this convention in this paper. In this case, for any fixed $n$, the price of fairness for that $n$ is the supremum over all instances with $n$ agents.

\subsection{Properties} \label{sec:prop}

The following section defines the various properties that allocations may satisfy. We will investigate the price of fairness of every one of them.

First, we define various fairness properties:

\begin{defn}[EF1]
An allocation $\calA$ is \textbf{envy-free up to one good (EF1)} if, for any pair of agents $i,j$, there exists $G \subseteq A_j$ with $|G| \le 1$ such that $u_i(A_i) \ge u_i(A_j \setminus G)$.
\end{defn}

\begin{defn}[Balanced]
An allocation $\calA$ is \textbf{balanced (Ba)} if, for any pair of agents $i,j$, we have $|A_i| - |A_j| \in \{-1, 0, 1\}$.
\end{defn}

\begin{defn}[RR]
The \textbf{round-robin algorithm} takes an instance $\calI$ and works as follows. First, it puts the agents in some order. Then, starting from the first agent and following the order, looping around whenever we reach the last agent, the algorithm assigns to an agent her most valuable good from those remaining. In case of a tie, the algorithm breaks ties arbitrarily.

An allocation is \textbf{round-robin (RR)} if it is produced by the round-robin algorithm, for some ordering of the agents and choices on tiebreaks.
\end{defn}

\begin{fact}\label{fact:rr}
A RR allocation is also EF1 and balanced~\cite{Caragiannis2019}. As a result, since RR is always satisfiable, EF1 and balancedness are also always satisfiable.
\end{fact}

We also define and investigate the following efficiency notions:

\begin{defn}[MUW]
The \textbf{utilitarian welfare} of an allocation $A$ is the sum of utilities $\sum_i u_i(A_i)$. The \textbf{maximum utilitarian welfare (MUW)} is the maximum possible utilitarian welfare; an allocation achieving that is also called MUW.
\end{defn}

\begin{defn}[MNW]
The \textbf{Nash welfare} of an allocation $A$ is the product of utilities $\prod_i u_i(A_i)$. The \textbf{maximum Nash welfare (MNW)} is the maximum possible Nash welfare; an allocation achieving that is also called MNW.
\end{defn}

\section{Fairness Properties}

In this section, we consider the price of fairness for the properties EF1, balanced, and RR. As mentioned in Fact~\ref{fact:rr}, a RR allocation is also EF1 and balanced, so these three properties are related. The results in this section are summarized in Table~\ref{tab:summary-EF1-bal-RR}.

\begin{table*}[ht]
\caption{Prices of EF1, balanced and RR}
\label{tab:summary-EF1-bal-RR}
\centering
\begin{tabular}{|l|c|c|}
    \hline
    \multicolumn{1}{|c|}{\multirow{2}{*}{Property}} & \multicolumn{2}{|c|}{Price of fairness} \\
    \cline{2-3}
    & {Lower Bound} & {Upper Bound} \\
    \hline
    Envy-free up to one good (EF1) & $n-1$ & $2n-1$ \\
    \hline
    Balanced & $n$ & $n$ \\
    \hline
    Round-robin algorithm (RR) & $n$ & $2n-1$ \\
    \hline
\end{tabular}
\end{table*}

We first provide a lower bound for the three properties.

\begin{thm} \label{thm:EF1lower}
$\POF_\EFi \ge n-1$ and $\POF_\RR, \POF_\Ba \ge n$.
\end{thm}

\begin{proof}
Let $m \gg n$ and $\eps \ll 1/m$. Consider the instance $\calI$ with following utilities:

\begin{itemize}
    \item $u_1(1) = 1$ and $u_1(j) = 0$ for $2 \le j \le m$.
    \item For $i = 2, \ldots, n-1$: $u_i(1) = 1 - (m-1) \eps$ and $u_i(j) = \eps$ for $2 \le j \le m$.
    \item $u_n(1) = 1 - (m-1) \eps^2$ and $u_n(j) = \eps^2$ for $2 \le j \le m$.
\end{itemize}

In any allocation with nonzero egalitarian welfare, agent 1 gets good 1, and each other agent gets at least one good. Once this is done, the minimum welfare is dictated by agent $n$. So, the optimal welfare is obtained by giving good $i$ to agent $i$ for $i = 1, \ldots, n-1$, and the remaining goods to agent $n$. This gives
$$\MEW = (m - (n-1)) \cdot \eps^2.$$

To obtain an EF1 allocation with nonzero welfare, agents $2, \ldots, n$ must split goods $2, \ldots, m$ as evenly as possible, giving
$$\MEW_\EFi = \left\lceil \frac{m-1}{n-1} \right\rceil \cdot \eps^2.$$

To obtain a balanced allocation with nonzero welfare, all agents must split the goods as evenly as possible, giving
$$\MEW_\Ba = \left\lceil \frac{m}{n} \right\rceil \cdot \eps^2.$$

Therefore, as $m \to \infty$,
$$\POF_\EFi(\calI) = \frac{\MEW}{\MEW_\EFi} \to n-1 \qquad \text{and} \qquad \POF_\Ba(\calI) = \frac{\MEW}{\MEW_\Ba} \to n.$$

This gives the lower bounds for EF1 and balancedness. For RR, note that any RR allocation is balanced, so $\MEW_\RR \le \MEW_\Ba$ and so $\POF_\RR \ge \POF_\Ba$.
\qed
\end{proof}

We now provide an upper bound proof for balancedness.

\begin{thm} \label{thm:BAupper}
$\POF_\Ba \le n$.
\end{thm}

\begin{proof}
Let the instance have $n$ agents and $m$ goods. Let $\lceil m/n \rceil = q$, and let the remainder of $m$ divided by $n$ be $r$; if $m$ is divisible by $n$, then $r = n$. Note that $m = n q - (n-r) \le n q$.

Given an instance, consider a MEW allocation $\calA$. For each agent $i$, let her keep the most valuable $q$ goods from her bundle $A_i$; if $|A_i| \le q$, then agent $i$ will get exactly $A_i$. However, if there are more than $r$ agents keeping $q$ goods, then only $r$ agents can keep $q$ goods; the rest can only keep $q-1$ goods. Pool the leftover goods and divide them arbitrarily such that the resulting allocation $\calB$ is balanced; the above guarantees such a balanced allocation exists.

For each agent $i$, there are three cases:

\begin{itemize}
    \item She had $\le q-1$ goods in $\calA$. Then she keeps all and so $u_i(B_i) \ge u_i(A_i)$.
    \item She has $q$ goods in $\calB$. Then she keeps the most valuable $q$ goods out of her initial bundle of $\le m$ goods, so $u_i(B_i) \ge \frac{q}{m} \cdot u_i(A_i) \ge \frac{1}{n} \cdot u_i(A_i)$.
    \item She has $q-1$ goods in $\calB$ and had $\ge q$ goods in $\calA$. Therefore, $r$ other agents have $q$ goods each; note that $r < n$. Then, agent $i$'s initial bundle had $\le m - r q = (n-r)(q-1)$ goods. Therefore, $u_i(B_i) \ge \frac{q-1}{(n-r)(q-1)} \cdot u_i(A_i) = \frac{1}{n-r} \cdot u_i(A_i)$.
\end{itemize}

In all cases, agent $i$'s utility in $\calB$ is at least $1/n$ of that in $\calA$. Therefore,
\begin{equation*}
    \POF_\Ba \le \frac{\EW(\calI, \calA)}{\EW(\calI, \calB)} \le n.
  \tag*{\qed}
\end{equation*}
\end{proof}

Combining Theorems~\ref{thm:EF1lower} and \ref{thm:BAupper}, we get the exact price for balancedness:

\begin{corl}
    $\POF_\Ba = n$.
\end{corl}

Next, we provide an upper bound for EF1 and RR.

\begin{thm} \label{thm:RRupper}
$\POF_\RR \le 2n-1$, and so, $\POF_\EFi \le 2n-1$.
\end{thm}

Before going to the proof, we establish some definitions.

\begin{defn}[Domination and Pareto-optimality]
An allocation $\calA$ is \textbf{weakly dominated} by an allocation $\calB$ if $u_i(A_i) \le u_i(B_i)$ for all agent $i$. It is \textbf{strongly dominated} by $\calB$ if, in addition, at least one of the inequalities is strict.

An allocation is \textbf{Pareto-optimal (PO)} if it is not strongly dominated by any allocation.
\end{defn}

\begin{defn}[Envy-Graph]
Given allocation $\calA$, its \textbf{envy-graph} $\EG_\calA$ is a directed graph defined as follows. The vertex set is the set of agents $N = \{1, 2, \ldots, n\}$. There is an edge $i \to j$ whenever $u_i(A_i) < u_i(A_j)$.
\end{defn}

We first prove the following property of the envy-graph of Pareto-optimal allocations, which will be useful in our proof of Theorem~\ref{thm:RRupper}.

\begin{lem} \label{lem:POacyclic}
The envy-graph $\EG_\calA$ of any Pareto-optimal allocation $\calA$ is acyclic.
\end{lem}

\begin{proof}
Suppose $\EG_\calA$ has a directed cycle $C$ of agents that envy the next in the cycle. Consider an allocation $\calA'$ that shifts each bundle in $C$ backward, so each agent receives the bundle she envied.

Note that agents not in $C$ retain their utilities, while agents in $C$ strictly improve their utilities. Therefore, $\calA'$ strongly dominates $\calA$. This contradicts that $\calA$ is Pareto-optimal. Therefore, such a cycle $C$ cannot exist, so $\EG_\calA$ is acyclic.
\qed
\end{proof}

Using the above lemma, we can now describe the ordering for the round-robin algorithm. In the simplified case with $m=n$, it yields a strong result.

\begin{lem} \label{lem:RRPO}
Consider an instance with $n$ agents and $m$ goods with $m = n$. Let $\calA$ be a balanced allocation (i.e. assigns one good to each agent). Then there exists an allocation $\calB$ that is produced by the round-robin algorithm for some ordering and tiebreaking mechanism, and that also weakly dominates $\calA$.
\end{lem}

\begin{proof}
Consider the set of balanced allocations that weakly dominate $\calA$. The set is non-empty (as $\calA$ is in it) and finite, so let $\calB$ be some Pareto-optimal allocation in this set. By Lemma~\ref{lem:POacyclic}, $\EG_\calB$ is acyclic. Therefore, the vertices admit a topological ordering $\pi$.

We now describe the settings for the round-robin algorithm. The ordering is $\pi^\text{rev}$, the reverse of the topological ordering we got. The tiebreaking mechanism is arbitrary except that each agent $i$ prefers the good $g_i$ assigned to her in $\calB$ compared to other goods of the same utility. We model this by increasing $u_i(g_i)$ slightly, such that agent $i$'s order of preference for the goods does not change. Since each agent only gets one good, the envy-graph $\EG_\calB$ is not affected.

We claim the round-robin algorithm, with the ordering and tiebreaking mechanism described above, produces $\calB$. Suppose it does not. Since $\calB$ is Pareto-optimal, there is an agent that is worse off; let $x$ be the earliest such agent in the round-robin ordering, and let $y$ be the agent that receives good $g_x$.

Agent $y$ must pick before $x$, otherwise $x$ would have been able to choose $g_x$ instead of her worse good. By choice of $x$, agent $y$ is not worse off, and so $u_y(g_x) \ge u_y(g_y)$. Equality cannot happen, since we adjusted the good utilities so $u_y(g_y)$ is not equal to anything else. So, $u_y(g_x) > u_y(g_y)$, and thus, $y$ envies $x$ in $\calB$. Then $y \to x$ is a directed edge in $\EG_\calB$, and so $y$ must appear before $x$ in $\pi$; this contradicts that $y$ picks before $x$ in the round-robin ordering $\pi^\text{rev}$.

Hence, no such agent $x$ exists, and the round-robin algorithm produces $\calB$.
\qed
\end{proof}

With the above lemmas, we are now ready to prove Theorem~\ref{thm:RRupper}.

\begin{proof}[Theorem~\ref{thm:RRupper}]
Consider an instance $\calI$; we may assume the optimal welfare is positive. Consider a MEW allocation $\calA = (A_1, A_2, \ldots, A_n)$. Since each agent has positive utility, their bundle is nonempty. For each agent $i$, let $g_i \in A_i$ be her most valuable good in her bundle.

Consider a reduced instance $\calI'$ that has the same $n$ agents, but only uses the goods $g_1, g_2, \ldots, g_n$. Let allocation $\calA'$ be the allocation that assigns good $g_i$ to agent $i$. Using Lemma~\ref{lem:RRPO}, there is an allocation $\calB'$ produced by a round-robin ordering $\pi$ that also weakly dominates $\calA'$. Let $g_i'$ be the good received by agent $i$ in $\calB'$; note that $u_i(g_i') \ge u_i(g_i)$ for all $i$.

We now use the same ordering $\pi$ to perform the round-robin algorithm over the initial instance $\calI$. The tiebreaking mechanism is the same: agent $i$ prefers good $g_i'$ if tied. Let $\calB = (B_1, B_2, \ldots, B_n)$ be the resulting allocation.

Let $x$ be an arbitrary agent. We will give a lower bound on $u_x(B_x) / u_x(A_x)$. Name the goods in bundle $A_x$ as $\{p_1, p_2, \ldots, p_\ell\}$ sorted in non-increasing utility; note that $p_1 = g_x$. Also, let agent $x$'s picks be $h_1, h_2, \ldots, h_k$ in order.

Consider agent $x$'s first pick $h_1$. All goods in $\calI'$ are present in $\calI$, so $u_x(h_1) \ge u_x(g_x')$ and so,
$$u_x(h_1) \ge u_x(g_x') \ge u_x(g_x) = u_x(p_1).$$

Consider agent $x$'s $k$-th pick $h_k$ for $k \ge 2$. At most $kn-1$ goods have been taken, so,
$$u_x(h_k) \ge u_x(p_{kn}).$$

Therefore,
\begin{align*}
u_x(A_x) = \sum_{i=1}^\ell u_x(p_i) &\le (2n-1) \cdot u_x(p_1) + n \cdot \sum_{k \ge 2} u_x(p_{kn}) \\
&\le (2n-1) \cdot u_x(h_1) + n \cdot \sum_{k \ge 2} u_x(h_k) \le (2n-1) \cdot u_x(B_x).
\end{align*}

In particular, let $x$ be an agent with minimum utility in $\calB$. Then, we have  $\EW(\calI, \calA) \leq u_x(A_x) \leq (2n-1) \cdot u_x(B_x) = (2n-1) \cdot \EW(\calI, \calB)$. Therefore,
\begin{equation*}
  \POF_\RR \le \frac{\EW(\calI, \calA)}{\EW(\calI, \calB)} \le 2n-1.
  \tag*{\qed}
\end{equation*}
\end{proof}

\section{Welfare Maximizers}

In this section, we consider the price of fairness for the properties MUW and MNW. We do not consider the price of fairness of MEW since it is 1 by definition. The results in this section are summarized in Table~\ref{tab:summary-MUW-MNW}.

\begin{table*}[ht]
\caption{Prices of MUW and MNW}
\label{tab:summary-MUW-MNW}
\centering
\begin{tabular}{|l|c|c|c|}
    \hline
    \multicolumn{2}{|c|}{\multirow{2}{*}{Property}} & \multicolumn{2}{|c|}{Price of fairness} \\
    \cline{3-4}
    \multicolumn{2}{|c|}{} & {Lower Bound} & {Upper Bound} \\
    \hline
    \multicolumn{2}{|l|}{Maximum utilitarian welfare (MUW)} & $\infty$ & $\infty$ \\
    \hline
    \multirow{2}{*}{Maximum Nash welfare (MNW)} & ($n = 2$) & 1.754$\ldots$ & 2 \\
    \cline{2-4}
    & ($n \ge 3$) & $\infty$ & $\infty$\\
    \hline
\end{tabular}
\end{table*}

We start with a result about MUW.

\begin{thm} \label{thm:MUW}
$\POF_\MUW = \infty$.
\end{thm}
\begin{proof}
Let $\eps \ll 1$. Take the instance with $n = 2$, $m = 3$ and the utilities below:

\begin{itemize}
    \item $u_1(1) = u_1(2) = 1/2$ and $u_1(3) = 0$.
    \item $u_2(1) = u_2(2) = 1/2 - \eps$ and $u_2(3) = 2 \eps$.
\end{itemize}

A MEW allocation is to assign good 1 to agent 1, and goods 2 and 3 to agent 2. The utilitarian welfare is $1 + \eps$, and the egalitarian welfare is $1/2$.

However, the MUW allocation is to assign goods 1 and 2 to agent 1, and good 3 to agent 2. The utilitarian welfare is $1 + 2 \eps$, and the egalitarian welfare is $2 \eps$. Thus the price of fairness is $(1/2) / (2 \eps) = 1/(4 \eps)$, which goes to $\infty$ as $\eps \to 0$.

To add more agents, we simply introduce $k$ new agents and $k$ new goods so that each new agent exclusively desires one of the new goods (with utility 1) without any overlap. Then the MEW and MUW are still obtained by assigning the new agents to the new goods, leaving only the original instance.
\qed
\end{proof}

For MNW, the behaviors of instances with $n = 2$ and instances with $n \ge 3$ differ substantially. We first provide a lower bound for the case of 2 agents.

\begin{thm} \label{thm:MNW2lower}
Consider instances with $n = 2$ agents. Let $\lambda = 1.324\ldots$ be the real number satisfying $\lambda^3 - \lambda - 1 = 0$. Then,

$$\POF_\MNW \ge \lambda^2 = 1.754\ldots$$
\end{thm}
\begin{proof}
Let $x, y$ be positive real numbers satisfying $x > 1$ and

\begin{equation} \label{eq:MNW2lower}
\frac{1}{x + \sqrt{x}} < y < \frac{1}{x^2}.
\end{equation}

Consider an instance $\calI$ with $n = 2$ and $m = 3$, with the following utilities:

\begin{itemize}
    \item $u_1(1) = xy$, and $u_1(2) = 1 - xy$, and $u_1(3) = 0$.
    \item $u_2(1) = 1 - xy$, and $u_2(2) = (x-1)y$, and $u_2(3) = y$.
\end{itemize}

There are three plausible allocations, presented in Table~\ref{tab:MNW2welfares} along with the Nash and egalitarian welfares. The allocations are labeled by which agent gets a good, e.g. 1-1-2 means good 1 goes to agent 1, good 2 also goes to agent 1, and good 3 goes to agent 2.

\begin{table}[ht]
\caption{Nash and egalitarian welfare of the instance}
\label{tab:MNW2welfares}
\centering
\begin{tabular}{|c|c|c|}
\hline
Allocation & Nash & Egalitarian \\
\hline
1-1-2 & $y$ & $y$ \\
\hline
1-2-2 & $x^2 y^2$ & $xy$ \\
\hline
2-1-2 & $(1 - xy) (1 - (x-1)y)$ & $1 - xy$ \\
\hline
\end{tabular}
\end{table}

Given that \eqref{eq:MNW2lower} holds, it can be verified that the MEW allocation is 1-2-2 but the MNW allocation is 1-1-2. Therefore, $\POF_\MNW(\calI) = xy / y = x$.

However, $y$ in \eqref{eq:MNW2lower} can exist only if the gap is non-empty, i.e.
$$\frac{1}{x + \sqrt{x}} < \frac{1}{x^2} \quad \Longleftrightarrow \quad (\sqrt{x})^3 < (\sqrt{x}) + 1.$$

Thus we require $\sqrt{x} < \lambda$, where $\lambda$ is the real number satisfying $\lambda^3 - \lambda - 1 = 0$; moreover, any such $x$ admits some $y$. The price of fairness is
\begin{equation*}
  \POF_\MNW \ge \sup x = \lambda^2 \approx 1.754\ldots.
  \tag*{\qed}
\end{equation*}

\end{proof}

We complement Theorem 5 by presenting an upper bound for MNW when $n = 2$.
\begin{thm} \label{thm:MNW2upper}
For $n = 2$ agents, $\POF_\MNW \le 2$.
\end{thm}
\begin{proof}
Let $\calI$ be an instance with $n = 2$ agents. Let $\calA_N = (N_1, N_2)$ be a MNW allocation and $\calA_E = (E_1, E_2)$ be a MEW allocation. We assume $\calA_N$ has strictly larger Nash welfare than $\calA_E$, otherwise there is nothing to prove. Note that positive Nash welfare implies no agent receives 0 utility.

Without loss of generality, suppose $u_1(E_1) \le u_2(E_2)$. We claim that $u_2(E_2) \ge 1/2$. Suppose that is not the case, then each agent values her bundle strictly less than 1/2. Swapping the bundles gives each agent a bundle with value strictly more than 1/2 and so improves the egalitarian welfare.

Let $x$ be the larger of $u_1(N_1)$ and $u_2(N_2)$, and $y$ be the smaller. The egalitarian welfare of $N$ is then $y$. We claim that $\frac{u_1(E_1)}{y} < 2$. Indeed, note that
$$xy > u_1(E_1) u_2(E_2) \quad \implies \quad \frac{x}{u_2(E_2)} > \frac{u_1(E_1)}{y}.$$

Now, note that $x \le 1$, since the utility of the whole set of goods is 1. Meanwhile, $u_2(E_2) \ge 1/2$ as proven above. Therefore,
$$\POF_\MNW(\calI) = \frac{u_1(E_1)}{y} < \frac{x}{u_2(E_2)} \le \frac{1}{1/2} = 2.$$

This works for any $\calI$, therefore $\POF_\MNW = \sup \POF_\MNW(\calI) \le 2$.
\qed
\end{proof}

Finally, we provide the price of MNW for the case of $n = 3$.

\begin{thm} \label{thm:MNW3}
For $n \ge 3$ agents, $\POF_\MNW = \infty$.
\end{thm}

\begin{proof}
Let $\eps \ll 1$. Take the instance with $n = m = 3$ and the following utilities:

\begin{itemize}
    \item $u_1(1) = 1$ and $u_1(2) = u_1(3) = 0$.
    \item $u_2(1) = 1/3 - \eps/2$, $u_2(2) = \eps/2$, and $u_2(3) = 2/3$.
    \item $u_3(1) = 1 - \eps/2 - \eps^2/2$, $u_3(2) = \eps^2/2$, and $u_3(3) = \eps/2$.
\end{itemize}

The MEW allocation is to assign good 1 to agent 1, good 2 to agent 2, and good 3 to agent 3. The Nash welfare is $\eps^2 / 4$, and the egalitarian welfare is $\eps/2$.

However, the MNW allocation is to assign good 1 to agent 1, good 2 to agent 3, and good 3 to agent 2. The Nash welfare is $\eps^2 / 3$, and the egalitarian welfare is $\eps^2/2$. Thus the price of fairness is $(\eps/2) / (\eps^2/2) = 1/\eps$, which goes to $\infty$ as $\eps \to 0$.

To add more agents, we simply introduce $k$ new agents and $k$ new goods so that each new agent exclusively desires one of the new goods (with utility 1) without any overlap. Then the MEW and MNW are still obtained by assigning the new agents to the new goods, leaving only the original instance.
\qed
\end{proof}

\section{Conclusion}

We extended the notion of price of fairness to a combination that has not been investigated yet in the literature: egalitarian welfare with indivisible goods. We found upper and lower bounds for the (egalitarian) price of fairness for several different fairness conditions: envy-free up to one good, round-robin, balanced, utilitarian welfare maximizing, and Nash welfare maximizing.

Similar to the results for utilitarian welfare found by Bei et al.~\cite{ijcai2019p0012} and Barman et al.~\cite{Barman2020}, our results establish the asymptotic growth exactly. In fact, for balancedness, we not only derive the asymptotic growth, but the exact growth down to the constant. We still have a multiplicative gap of 2 between the bounds for EF1 and RR.

For welfare maximizers, it turns out that in many of the cases, maximizing any other welfare can come at arbitrarily large cost for the egalitarian welfare. The exception is maximizing Nash welfare with 2 agents, for which we have a finite price of fairness. There is still an unresolved gap between the lower and upper bounds.

Besides tightening the bounds, other directions for future work are to extend the results to other properties not discussed in this paper, and to investigate other fair division settings, such as using chores instead of goods.

It is also possible to stay within the realm of fair division of indivisible goods, but with other kinds of welfare. Bei et al.~\cite{ijcai2019p0012} and later Barman et al.~\cite{Barman2020} have investigated the case of utilitarian welfare, and we have considered egalitarian welfare, so the obvious next step is to look at Nash welfare. More generally, it is also possible to use the generalized H\"older mean introduced by Arunachaleswaran et al.~\cite{arunachaleswaran2021fair} that interpolates between these three kinds of welfare.

Bei et al.~\cite{ijcai2019p0012} also introduced the concept of \textit{strong price of fairness} which represents efficiency loss in the worst fair allocation instead of in the best fair allocation. One possible direction would be to study the strong price of fairness with respect to egalitarian welfare or other welfare measures.

\subsubsection{Acknowledgements} The authors would like to thank their lecturer Warut Suksompong for his valuable contributions.

%
%
%
\bibliographystyle{splncs04}
\bibliography{arxiv}
\end{document}